\documentclass[10pt,conference,letterpaper]{IEEEtran}
\usepackage{fixltx2e}
\usepackage{cite}
\usepackage{url}
\usepackage{mathrsfs}
\usepackage{color}
\usepackage{float}
\usepackage[caption=false]{subfig}

\ifCLASSINFOpdf
   \usepackage[pdftex]{graphicx}
   \graphicspath{{Figs/}}
   \DeclareGraphicsExtensions{.pdf,.jpeg,.png}
\else
\fi
\usepackage{balance}
\usepackage[cmex10]{amsmath}
\usepackage{amsmath}
\usepackage{amssymb}
\usepackage{amsthm}
\usepackage{amsfonts}
\usepackage{bm}
\usepackage{xfrac}
\usepackage{empheq}
\usepackage[normalem]{ulem} 
\usepackage{soul} 
\usepackage{mathtools}

\newtheorem{theorem}{Theorem}

\newtheorem{example}{Example}
\newtheorem{proposition}{Proposition}
\newtheorem{lemma}{Lemma}

\newtheorem{corollary}{Corollary}
\newtheorem{remark}{Remark}
\theoremstyle{definition}

\usepackage[a4paper,bindingoffset=0.2in,%
left=1in,right=1in,top=1in,bottom=1in,%
footskip=.25in]{geometry}
\graphicspath{{figs/}}

\begin{document}
	\newgeometry{left=0.7in,right=0.7in,top=.5in,bottom=1in}
	\title{An Information Geometric Approach to Local Information Privacy with Applications to Max-lift and Local Differential Privacy}
\vspace{-4mm}
\author{
	\IEEEauthorblockN{Amirreza Zamani$^\dagger$, Parastoo Sadeghi$^\ddagger$, Mikael Skoglund$^\dagger$ \vspace*{0.5em}
		\IEEEauthorblockA{\\
			$^\dagger$Division of Information Science and Engineering, KTH Royal Institute of Technology \\
			$^\ddagger$School of Engineering and Technology, UNSW\\
			Email: \protect amizam@kth.se, p.sadeghi@unsw.edu.au, skoglund@kth.se }}
}
	\maketitle

\begin{abstract}
	We study an information-theoretic privacy mechanism design, where an agent observes useful data $Y$ and wants to reveal the information to a user. Since the useful data is correlated with the private data $X$, the agent uses a privacy mechanism to produce disclosed data $U$ that can be released. We assume that the agent observes $Y$ and has no direct access to $X$, i.e., the private data is hidden. We study the privacy mechanism design that maximizes the revealed information about $Y$ while satisfying a bounded Local Information Privacy (LIP) criterion.
	When the leakage is sufficiently small, concepts from information geometry allow us to locally approximate the mutual information. By utilizing this approximation the main privacy-utility trade-off problem can be rewritten as a quadratic optimization problem that has closed-form solution under some constraints. For the cases where the closed-form solution is not obtained we provide lower bounds on it. In contrast to the previous works that have complexity issues, here, we provide simple privacy designs with low complexity which are based on finding the maximum singular value and singular vector of a matrix. To do so, we follow two approaches where in the first one we find a lower bound on the main problem and then approximate it, however, in the second approach we approximate the main problem directly.
	
	 In this work, we present geometrical interpretations of the proposed methods and in a numerical example we compare our results considering both approaches with the optimal solution and the previous methods. Furthermore, we discuss how our method can be generalized considering larger amounts for the privacy leakage. Finally, we discuss how the proposed methods can be applied to deal with differential privacy.	 
\end{abstract}
\section{Introduction}
As shown in Fig.~\ref{fig:sysmodel1}, in this paper, an agent tries to reveal some useful information to a user. Random variable (RV) $Y$ denotes the useful data and is arbitrarily correlated with the private data denoted by RV $X$. Furthermore, RV $U$ describes the disclosed data. The agent wants to design $U$ based on $Y$ that reveals as much information as possible about $Y$ and satisfies a privacy criterion. We use mutual information to measure utility and Local Information Privacy (LIP) to measure the privacy leakage. In this work, some bounded privacy leakage is allowed, i.e., for all $x$ and $u$ we require $-\epsilon\leq \log(\frac{P_{X|U}(x|u)}{P_{X}(x)})\leq \epsilon$.

Related works on the statistical privacy mechanism design can be found in
\cite{khodam,Khodam22,zarab1,zarab2,zarab3,seif,duchi,kairouz,evfimievski,kairouz2015composition,barthe2013beyond,feldman2018privacy,koala,borz,shah, makhdoumi, dwork1,dwork2006our, shahab,sankar, Total, sankar2, deniz4, asoodeh3, Calmon1,  nekouei2, issa,oof,razegh,emma,Liu,multiAmir,sep,lopuha}. 
 

In \cite{borz}, the problem of privacy-utility trade-off considering mutual information both as measures of privacy and utility is studied. Under perfect privacy assumption, it has been shown that the privacy mechanism design problem can be reduced to linear programming. 
In \cite{khodam}, privacy mechanisms with a per letter privacy criterion considering an invertible leakage matrix have been designed allowing a bounded leakage. This result is generalized to a non-invertible leakage matrix in \cite{Khodam22}.
\begin{figure}[]
	\centering
	\includegraphics[width = 0.35\textwidth]{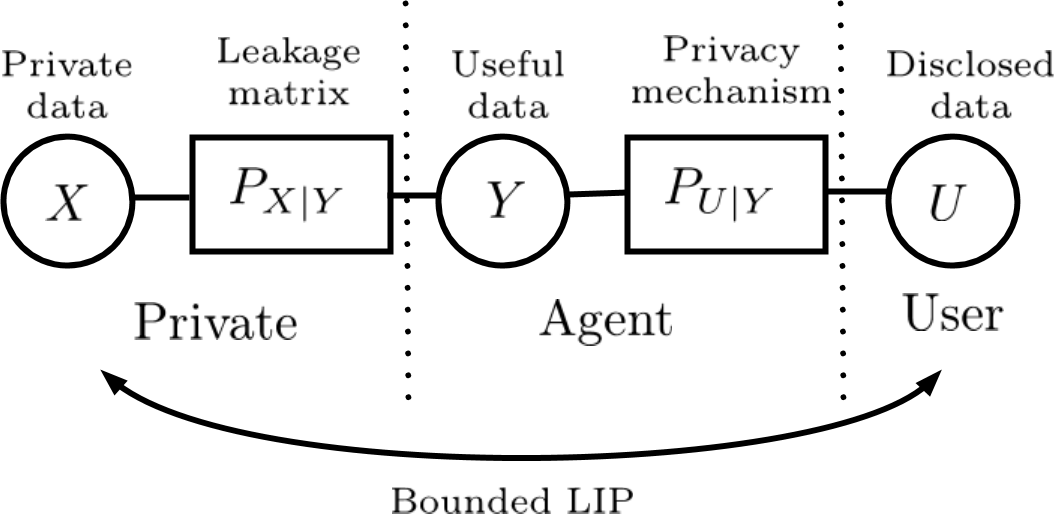}
	\caption{In this model, disclosed data $U$ is designed by a privacy mechanism that maximizes the information disclosed about $Y$ and satisfies the bounded LIP criterion. Here, we assume that the private data $X$ is not available directly to the agent.}
	\label{fig:sysmodel1}
\end{figure}
In \cite{koala}, \emph{secrecy by design} problem is studied under the perfect secrecy assumption. Bounds on secure decomposition have been derived using the Functional Representation Lemma. 
In \cite{shah}, the privacy problems considered in \cite{koala} are generalized by relaxing the perfect secrecy constraint and allowing some leakage. Furthermore, the bounds in \cite{shah} have been tightened in \cite{sep} by using a \emph{separation technique}. In \cite{Liu}, fundamental limits of private data disclosure are studied, where the goal is to minimize leakage under utility constraints with non-specific tasks. This result is generalized in \cite{multiAmir}.  
The concept of lift is studied
in \cite{zarab2} which represents the likelihood ratio between the posterior and prior beliefs
concerning sensitive features within a data set.
Concepts from information geometry have been used in \cite{khodam,Khodam22,shah,razegh}, and \cite{emma}, to approximate the main design problems and find simple privacy designs. Specifically, the \emph{strong $\chi^2$-privacy criterion} and the \emph{strong $\ell_1$-privacy criterion} are introduced in \cite{khodam} and \cite{Khodam22}. Both per-letter (point-wise) measures let us approximate the main privacy-utility trade-off problems and study them geometrically. Furthermore, in \cite{emma}, by using information geometry a local approximation of the secrecy capacity over a wire-tap channel has been obtained.  

In the context of privacy, numerous measures of privacy leakage have been used, e.g., mutual information \cite{borz,shah,koala,Liu}, differential privacy (DP) \cite{dwork11}, $(\epsilon,\delta)$-differential privacy \cite{dwork2006our}, local differential privacy (LDP) \cite{kairouz2015composition,cuff2016differential,barthe2013beyond,feldman2018privacy,duchi,kairouz,evfimievski,lopuha}, local information privacy (LIP) \cite{zarab3,Calmon1,seif,zarab1,lopuha}, maximal leakage \cite{issa}, lift \cite{zarab2}, average $\ell_1$-distance \cite{Total}, average $\chi^2$-divergence \cite{Calmon2}, point-wise $\ell_1$-distance \cite{Khodam22}, point-wise $\chi^2$-divergence \cite{khodam,razegh}.

Many information theory problems face challenges due to the lack of a geometric structure in the space of probability distributions. Assuming the distributions of interest are close, KL divergence as well as mutual information can be approximated by a weighted squared Euclidean distance. This leads to a method where we can approximate the main problems. This approach has been used in \cite{Shashi, huang}, considering point-to-point channels and some specific broadcast channels. As we outlined earlier, a similar approach has been used in the privacy context in \cite{khodam,Khodam22,shah,razegh}, and \cite{emma}.

In the present work, due to the bounded LIP criterion which is a point-wise (strong) measure, we utilize concepts from the information geometry similar as \cite{khodam}, to approximate the KL divergence and mutual information in case of a small leakage $\epsilon$. This allows us to transfer the main problem into an analytically simple quadratic linear algebra problem, which also provides deep intuitive understanding of the mechanism. To do so, we follow two approaches. In the first approach, we first find a lower bound on the main privacy-utility trade-off problem using existing inequalities on $\log(1+x)$, then we approximate it using KL-approximation. In the second approach, we directly find an approximation of the main problem which leads to a quadratic optimization (linear algebra) problem. We compare the obtained results with the optimal solution and previous results in \cite{khodam} in a numerical example. Furthermore, we present a geometrical interpretation of the proposed approach. Finally, we discuss how we can generalize the results for larger amount of privacy leakage and how the proposed approach can be applied for privacy mechanism designs considering LIP for the leakage constraints.  

\vspace{-1mm}
\section{System model and Problem Formulation} \label{sec:system}
Let $P_{XY}$ denote the joint distribution of discrete random variables $X$ and $Y$ defined on finite alphabets $\cal{X}$ and $\cal{Y}$ with equal cardinality, i.e, $|\cal{X}|=|\cal{Y}|=\mathcal{K}$.
We represent $P_{XY}$ by a matrix defined on $\mathbb{R}^{|\mathcal{K}|\times|\mathcal{K}|}$ and 
marginal distributions of $X$ and $Y$ by vectors $P_X$ and $P_Y$ defined on $\mathbb{R}^{|\mathcal{K}|}$ and $\mathbb{R}^{|\mathcal{K}|}$ given by the row and column sums of $P_{XY}$. 
We assume that each element in vectors $P_X$ and $P_Y$ is non-zero. Furthermore, 
we represent the leakage matrix $P_{X|Y}$ by a matrix defined on $\mathbb{R}^{|\mathcal{K}|\times|\mathcal{K}|}$, which is assumed to be invertible. Furthermore, for given $u\in \mathcal{U}$, $P_{X,U}(\cdot,u)$ and $P_{X|U}(\cdot|u)$ defined on $\mathbb{R}^{|\mathcal{X}|}$ are distribution vectors with elements $P_{X,U}(x,u)$ and $P_{X|U}(x|u)$ for all $x\in\cal X$ and $u\in \cal U$. 
The relation between $U$ and $Y$ is described by the kernel $P_{U|Y}$ defined on $\mathbb{R}^{|\mathcal{U}|\times|\mathcal{Y}|}$, furthermore, the relation between $U$ and the pair $(Y,X)$ is described by the kernel $P_{U|Y,X}$ defined on $\mathbb{R}^{|\mathcal{U}|\times|\mathcal{Y}|\times|\mathcal{X}|}$. In this work, $P_{X}(x)$, $P_{X}$ and $[P_{X}]$ denote $P_{X}(X=x)$, distribution vector of $X$ and a diagonal matrix with diagonal entries equal to $P_{X}(x)$, respectively. For two vectors $P$ and $Q$ with same size, we say $P\leq Q$ if $P(x)\leq Q(x)$ for all $x$.

Our goal is to design the privacy mechanism that produces the disclosed data $U$, which maximizes the utility and satisfies a privacy criterion.
In this work, the utility is measured by the mutual information $I(U;Y)$ and the privacy leakage by the LIP. Thus, the privacy problem can be stated as follows 
\begin{subequations}\label{problem}
	\begin{align}
	\sup_{P_{U|Y}} \ \ &I(U;Y),\label{problem1}\\
	\text{subject to:}\ \ &X-Y-U,\label{Markov1}\\
	 -\epsilon\leq& \log(\frac{P_{X|U}(x|u)}{P_{X}(x)})\leq \epsilon, \forall x,u.\label{local1}
	\end{align}
\end{subequations}
Intuitively, for small $\epsilon$, \eqref{local1} means that the two distributions (vectors) $P_{X|U=u}$ and $P_X$ are close to each other. This should hold for all $u\in\mathcal{U}$. Thus $X$ and $U$ are almost independent in the sense that $P_{X|U=u}$ almost does not depend on $U$.
The closeness of $P_{X|U=u}$ and $P_X$ allows us to locally approximate $I(U;Y)$ which leads to an approximation of \eqref{problem}. In the literature, The LIP constraint is based on $\log(\frac{P_{U|X}(u|x)}{P_{U}(u)})$, however, we can replace it by $\log(\frac{P_{X|U}(x|u)}{P_{X}(x)})$ since we have $\frac{P_{X|U}(x|u)}{P_{X}(x)}=\frac{P_{U|X}(u|x)}{P_{U}(u)}$. Furthermore, LIP is a well-known measure, and its relation with other measures can be found in the literature. For instance, \cite[Lemma 1]{lopuha} finds the relations between LIP, LDP, and mutual information. Additionally, removing the left inequality in \eqref{problem} results in the \emph{max-lift} privacy leakage measure \cite{zarab1, zarab2}. The relations between max-lift, the \emph{strong $\ell_1$-privacy criterion} and the \emph{strong $\chi^2$-privacy criterion} are studied in \cite{zarab1}.
\begin{remark}
	In this work, we assume that $P_{X|Y}$ is invertible; however, this assumption can be generalized using the techniques as in \cite{Khodam22}. Here, we focus on the invertible case for $P_{X|Y}$ and we discuss how to extend it in Section \ref{dis2}.
\end{remark}
\begin{remark}
	To solve \eqref{problem}, a linear program is proposed in \cite{lopuha} but with complexity issues as the size of $X$ and $Y$ grow. The linear program is based on finding extreme points of a set and trying all possible candidates in a two-step optimization problem which leads to an exponential computational complexity. In contrast to \cite{lopuha}, here, we propose a method that find lower bounds and approximations which are based on finding the maximum singular value and vector of a matrix and lead to simple privacy designs with intuitive geometrical interpretations. 
\end{remark}
\section{preliminaries and background}
In this section, we present the method used in \cite{Shashi,huang} and \cite{khodam} which allows us to approximate the mutual information. Using \eqref{local1}, we can rewrite the conditional distribution $P_{X|U=u}$ as a perturbation of $P_X$. Thus, for any $u\in\mathcal{U}$, we can write $P_{X|U=u}=P_X+\epsilon\cdot J_u$, where $J_u\in\mathbb{R}^\mathcal{K}$ is a perturbation vector that has the following three properties
\begin{align}
\sum_{x\in\mathcal{X}} J_u(x)&=0,\ \forall u,\label{prop1}\\
\sum_{u\in\mathcal{U}} P_U(u)J_u(x)&=0,\ \forall x\label{prop2},\\
\left(\frac{e^{-\epsilon}-1}{\epsilon}\right)P_{X}(x)\leq J_{u}(x)&\leq \left(\frac{e^{\epsilon}-1}{\epsilon}\right)P_{X}(x), \forall u, \ \forall x\label{prop3}.
\end{align}
The first two properties ensure that $P_{X|U=u}$ is a valid probability distribution \cite{khodam,Khodam22}, and the third property follows from \eqref{local1}.

Next, we recall a result from \cite{khodam}. To do so, let $[\sqrt{P_Y}^{-1}]$ and $[\sqrt{P_X}]$ be diagonal matrices with diagonal entries $\{\sqrt{P_Y}^{-1},\ \forall y\in\mathcal{Y}\}$ and $\{\sqrt{P_X},\ \forall x\in\mathcal{X}\}$. Furthermore, let $L_u\triangleq[\sqrt{P_X}^{-1}]J_u\in\mathbb{R}^{\mathcal{K}}$ and $W\triangleq [\sqrt{P_Y}^{-1}]P_{X|Y}^{-1}[\sqrt {P_X}]$. Finally, we use the Bachmann-Landau notation where $o(\epsilon)$ describes the asymptotic behaviour of a function $f:\mathbb{R}^+\rightarrow\mathbb{R}$ which satisfies $\frac{f(\epsilon)}{\epsilon}\rightarrow 0$ as $\epsilon\rightarrow 0$.
\begin{proposition}\cite[Proposition 3]{khodam}
	For a small $\epsilon$, $I(U;Y)$ can be approximated as follows
	\begin{align}
	I(Y;U)&=\frac{1}{2}\epsilon^2\sum_u P_U\|WL_u\|^2+o(\epsilon^2)\label{approx1}\\&\cong\frac{1}{2}\epsilon^2\sum_u P_U\|WL_u\|^2,\label{approx2}
	\end{align}
	where $\|\cdot\|$ corresponds to the Euclidean norm ($\ell_2$-norm).
\end{proposition}
\begin{proof}
	The proof is based on local approximation of the KL-divergence and can be found in \cite{khodam}.
\end{proof}
The next result recalls a property of matrix $W$.
\begin{proposition}\cite[Appendix C]{khodam}
	The smallest singular value of $W$ is $1$ with corresponding singular vector $\sqrt{P_{X}}$.
\end{proposition}
Finally, using \cite{khodam}, we recall that \eqref{prop1} can be rewritten as the constraint where vectors $\sqrt{P_{X}}$ and $L_u$ are orthogonal, i.e., 
\begin{align}\label{c1}
L_u \perp \sqrt{P_{X}},
\end{align}
\eqref{prop2} can be replaced by 
\begin{align}\label{c2}
\sum_{u\in\mathcal{U}} P_U(u)L_u=\bm 0\in\mathbb{R}^{\mathcal{K}},
\end{align}
and by using $L_u=[\sqrt{P_X}^{-1}]J_u$, \eqref{prop3} can be rewritten as
\begin{align}\label{c3}
\left(\frac{e^{-\epsilon}-1}{\epsilon}\right)\sqrt{P_{X}}\leq L_u \leq \left(\frac{e^{\epsilon}-1}{\epsilon}\right)\sqrt{P_{X}},\ \forall x,u.
\end{align}
 \section{Main Results}\label{sec:resul}
 \vspace{-1mm}
In this part, we derive lower bounds and approximations of \eqref{problem}. To do so, we follow two approaches. In the first approach, we find a lower bound on \eqref{problem} using lower and upper bounds on $\log(1+x)$ and then approximate it. In the second approach, we directly approximate \eqref{problem}.
\subsection{First Approach: Lower Bounds on \eqref{problem}}
In this part, we find lower bounds on \eqref{problem}. To do so, we use lower and upper bounds on $\log(1+x)$ which implies a stronger privacy criterion compared to the LIP in \eqref{local1}. 
\begin{lemma}\label{lem1}
	For all $\epsilon<1$, let $J_u$ satisfy \eqref{prop1}, and \eqref{prop2}, and 
	\begin{align}\label{jadid}
	\frac{-P_X(x)}{1+\epsilon}\leq J_u(x)\leq P_X(x), \ \forall x.
	\end{align}
	Then, \eqref{jadid} implies LIP in \eqref{local1}.
\end{lemma}
\begin{proof}
	We have
	\begin{align*}
	\log(\frac{P_{X|U}(x|u)}{P_{X}(x)})&=\log(1+\epsilon\frac{J_u(x)}{P_X(x)})\\
	&\stackrel{(a)}{\leq} \epsilon\frac{J_u(x)}{P_X(x)}\\&\stackrel{(b)}{\leq} \epsilon,
	\end{align*}
	where (a) follows by $\log(1+x)\leq x$ for all $x>-1$, and (b) follows by the right hand side in \eqref{jadid}. Furthermore, we have
	\begin{align*}
	\log(\frac{P_{X|U}(x|u)}{P_{X}(x)})&=\log(1+\epsilon\frac{J_u(x)}{P_X(x)})\\
	&\stackrel{(a)}{\geq} \epsilon\frac{J_u(x)}{P_X(x)+\epsilon J_u(x)}\\& \stackrel{(b)}{\geq} -\epsilon,
	\end{align*}
	where (a) follows by $\log(1+x)\geq \frac{x}{x+1}$ for all $x>-1$, and (b) follows by the left hand side in \eqref{jadid}. We emphasize that in (b) we used the fact that $P_X(x)+\epsilon J_u(x)\geq 0$ since $P_{X|U}(x|u)\geq 0$.
\end{proof}
In the next result we present a lower bound on \eqref{problem}.
\begin{corollary}\label{cor1}
	For all $\epsilon<1$, we have
	\begin{align}\label{app1lower}
\max_{\begin{array}{c} 
	\substack{P_{U|Y}: X-Y-U,\\ -\epsilon\leq \log(\frac{P_{X|U}(x|u)}{P_{X}(x)})\leq \epsilon, \forall x,u}
	\end{array}} \!\!\!\!\!\!\!\!\!\!\!\!\!\!\!\!I(Y;U)\geq \max_{\begin{array}{c} 
		\substack{J_u,\ P_U: X-Y-U,\\ J_u\ \text{satisfies}\ \eqref{prop1},\ \eqref{prop2},\ \text{and}\ \eqref{jadid}}
		\end{array}} \!\!\!\!\!\!\!\!\!\!\!\!\!\!\!\!I(Y;U).
	\end{align} 
\end{corollary}
\begin{proof}
	The proof follows by Lemma \ref{lem1}, since, \eqref{jadid} leads to the bounded LIP constraint. Moreover, $\epsilon<1$, \eqref{prop1} and \eqref{prop2} ensures that $P_X+\epsilon J_u=P_{X|U}$ is a distribution vector.
\end{proof}
\begin{remark}
	As $\epsilon$ decreases, the lower bound in \eqref{app1lower} becomes tighter. This follows since the upper and lower bounds on $\log(1+x)$ are obtained by using Taylor expansion and as $x$ decreases the error term becomes smaller.
\end{remark}
Next, we approximate the lower bound in \eqref{app1lower} using \eqref{approx1} and \eqref{approx2}. In the next result, $\alpha_{ij}$ corresponds to the $(i,j)$-th element of the matrix $P_{X|Y}^{-1}$. 
\begin{proposition}\label{pos3}
	For all $\epsilon<\max\{c_1,c_2\}$ and invertible leakage matrix $P_{X|Y}$, we have 
	\begin{align*}
	\max_{\begin{array}{c} 
		\substack{J_u,P_U: X-Y-U,\\ J_u\ \text{satisfies}\ \eqref{prop1},\ \eqref{prop2},\ \text{and}\ \eqref{jadid}}
		\end{array}} \!\!\!\!\!\!\!\!\!\!\!\!\!\!\!\!I(Y;U)=P_1+o(\epsilon^2)\cong P_1.
	\end{align*}
	where 
	\begin{align}\label{lower}
	P_1\triangleq\max_{\begin{array}{c} 
		\substack{L_u,P_U: \frac{-\sqrt{P_X}}{1+\epsilon}\leq L_u\leq \sqrt{P_X},\forall u ,\\ L_u\ \text{and}\ P_U\ \text{satisfy}\ \eqref{c1},\ \text{and}\ \eqref{c2}}
		\end{array}} \!\!\!\!\!\!\!\!\!\!\!\!\!\!\!\!\!\!\!0.5\epsilon^2\!\!\left(\sum_u \!\!P_U\|WL_u\|^2 \right),
	\end{align}
	and $c_1 \triangleq \frac{\min_y P_Y(y)}{\max_i \left(\sum_{j=1}^{\mathcal{K}}|\alpha_{ij}|P_X(x) \right)}$, $c_2\triangleq|\sigma_{\text{min}}(P_{X|Y})|\left(\min_y P_Y(y)\right)$ and $[\alpha_{ij}]_{\{1\leq i,j\leq \mathcal{K}\}}\triangleq P_{X|Y}^{-1}$.
\end{proposition}
\begin{proof}
	As we outlined earlier $I(U;Y)$ can be rewritten as $\frac{1}{2}\epsilon^2\sum_u P_U\|WL_u\|^2+o(\epsilon^2)$ for an invertible leakage matrix $P_{X|Y}$. 
	We emphasize that to approximate $I(U;Y)$, using the Markov chain $X-Y-U$ and invertible $P_{X|Y}$ we can rewrite $P_{Y|U=u}$ as perturbations of $P_Y$ as follows
	\begin{align*}
	P_{Y|U=u}&=P_{X|Y}^{-1}[P_{X|U=u}-P_X]+P_Y\\&=\epsilon\cdot P_{X|Y}^{-1}J_u+P_Y.
	\end{align*}
	Then, by using local approximation of the KL-divergence which is based on the second order Taylor expansion of $\log(1+x)$ we get
	\begin{align*}
	I(Y;U)&=\sum_u P_U(u)D(P_{Y|U=u}||P_Y)\\&=\sum_u P_U(u)\sum_y\! P_{Y|U=u}(y)\log\!\!\left(\!1\!+\!\epsilon\frac{P_{X|Y}^{-1}J_u(y)}{P_Y(y)}\right)\\&=\frac{1}{2}\epsilon^2\sum_u P_U\sum_y
	\frac{(P_{X|Y}^{-1}J_u)^2}{P_Y}+o(\epsilon^2)\\
	&=\frac{1}{2}\epsilon^2\sum_u P_U\|WL_u\|^2+o(\epsilon^2).
	\end{align*}
	For more details see \cite[Proposition 3]{khodam}.
	Furthermore, \eqref{prop1} and \eqref{prop2} are replaced by \eqref{c1} and \eqref{c2}. By using $L_u=[\sqrt{P_X}^{-1}]J_u$, \eqref{jadid} can be rewritten as 
	\begin{align*}
	\eqref{jadid}\leftrightarrow \frac{-\sqrt{P_X}}{1+\epsilon}\leq L_u\leq \sqrt{P_X}.
	\end{align*}
	For approximating $I(U;Y)$, we use the second Taylor expansion of $\log(1+x)$. Therefore, we must have $|\epsilon\frac{P_{X|Y}^{-1}J_u(y)}{P_Y(y)}|<1$ for all $u$ and $y$. One sufficient condition for $\epsilon$ to satisfy this inequality is to have $\epsilon<|\sigma_{\text{min}}(P_{X|Y})|\left(\min_y P_Y(y)\right)=c_2$, since in this case we have
	\begin{align*}
	\epsilon^2|P_{X|Y}^{-1}J_u(y)|^2&\leq\epsilon^2\left\lVert P_{X|Y}^{-1}J_u\right\rVert^2\leq\epsilon^2 \sigma_{\max}^2\left(P_{X|Y}^{-1}\right)\left\lVert J_u\right\rVert^2\\&\stackrel{(a)}\leq\frac{\epsilon^2}{\sigma^2_{\text{min}}(P_{X|Y})}<\min_{y\in\mathcal{Y}} P_Y^2(y),
	\end{align*}
	which implies $|\epsilon\frac{P_{X|Y}^{-1}J_u(y)}{P_Y(y)}\!|<1$. The step (a) follows from \eqref{jadid}, since \eqref{jadid} implies $\|J_u\|^2\leq \sum_x P_X(j)= 1$. Furthermore, another sufficient condition for $\epsilon$ to satisfy the inequality $|\epsilon\frac{P_{X|Y}^{-1}J_u(y)}{P_Y(y)}|<1$ is to have $\epsilon<\frac{\min P_Y(y)}{\max_i \left(\sum_{j=1}^{\mathcal{K}}|\alpha_{ij}|P_X(i) \right)}=c_1$, since in this case we have
	\begin{align*}
	\epsilon|P_{X|Y}^{-1}J_u(y)|&\stackrel{(a)}{\leq} \epsilon \sum_x |\alpha_{yx}||J_u(x)|\leq \epsilon \sum_x |\alpha_{yx}|P_X(x)\\ &< \min_y P_Y(y) \frac{\sum_x |\alpha_{yx}|P_X(x)}{\max_y \left(\sum_{x=1}^{\mathcal{K}}|\alpha_{yx}|P_X(x) \right)}\\&\leq \min_y P_Y(y),
	\end{align*}
	where (a) follows by $P_{X|Y}^{-1}J_u(y)=\sum_x \alpha_{yx}J_u(x)$.
\end{proof}
\begin{remark}
	In Section \ref{dis3}, we discuss how the range of  $\epsilon$ can be extended. 
\end{remark}
Next, we find lower and upper bounds on $P_1$. The lower bounds are shown to be optimal up to a constant scaling factor. To do so, let us define $\sigma_{\text{max}}$ and $L^*$ as the maximum singular value of matrix $W$ and corresponding singular vector, respectively. Furthermore, we assume $\|L^*\|=1$, otherwise we scale it. As we argued before, $\sigma_{\text{max}}>1$ and $L^*\perp \sqrt{P_X}$, e.g., see \cite[Appendix C]{khodam}.
\begin{lemma}\label{lem2}
Let $\gamma_1\geq1$ and $\gamma_2\geq1$ be the smallest scaling factors which ensure that $\frac{L^*}{\gamma_1}$ and $-\frac{L^*}{\gamma_2}$ satisfy $\frac{-\sqrt{P_X}}{1+\epsilon}\leq L\leq \sqrt{P_X}$, respectively. In other words, we divide $L^*$ by the smallest possible number $\gamma_1\geq1$ to ensure that the privacy constraint holds. Furthermore, let $\gamma_{\max}$ correspond to the smallest scaling factor to ensure the feasibility of $\frac{L^*}{\gamma_{\max}}$ and $-\frac{L^*}{\gamma_{\max}}$. In other words, to find the second lower bounds, we divide both $L^*$ and $-L^*$ with the same scaling factor $\gamma_{\max}$. We have
\begin{align}
\frac{1}{2}\epsilon^2\frac{\sigma_{\text{max}}^2}{\gamma_{\max}^2}\leq \frac{1}{2}\epsilon^2\frac{\sigma_{\text{max}}^2}{\gamma_1\gamma_2}\leq P_1\leq \frac{1}{2}\epsilon^2\sigma_{\text{max}}^2,
\end{align}  
Finally, for $|\mathcal{X}|=|\mathcal{Y}|=\mathcal{K}=2$, we have
\begin{align}
P_1=\frac{1}{2}\epsilon^2\frac{\sigma_{\text{max}}^2}{\gamma_1\gamma_2}.\label{cs}
\end{align}
\end{lemma}
\begin{proof}
	The proof is provided in Appendix A.
\end{proof}
Next, we present a result on the feasibility of $L^*$ and $-L^*$.
\begin{proposition}\label{prop4}
	Both $L^*$ and $-L^*$, with $\|L^*\|=1$, do not simultaneously satisfy the privacy constraint $\frac{-\sqrt{P_X}}{1+\epsilon}\leq L\leq \sqrt{P_X}$ in \eqref{jadid}. In other words, we have $\gamma_1>1$ and $\gamma_2\geq 1$, or $\gamma_1\geq1$ and $\gamma_2> 1$.
\end{proposition}
\begin{proof}
	Assume that $L^*$ and $-L^*$ satisfy the privacy constraint \eqref{jadid}. we have
	\begin{align}
	\frac{-\sqrt{P_X}}{1+\epsilon}\leq L^*\leq \sqrt{P_X}\label{d}\\
	\frac{-\sqrt{P_X}}{1+\epsilon}\leq -L^*\leq \sqrt{P_X}.\label{dd}
	\end{align}
	Multiplying \eqref{d} by $-1$ and combining it with \eqref{dd} we get
	\begin{align}\label{ddd}
	\frac{-\sqrt{P_X}}{1+\epsilon}\leq -L^*\leq 	\frac{\sqrt{P_X}}{1+\epsilon}. 
	\end{align}
	Hence, using \eqref{ddd} we obtain
	\begin{align*}
	\|L^*\|^2\leq \frac{1}{\left(1+\epsilon\right)^2},
	\end{align*}
	which contradicts $\|L^*\|=1$.
\end{proof}
By using Corollary \ref{cor1}, Proposition \ref{pos3} and Lemma \ref{lem2} we obtain the following result.
\begin{theorem}\label{th1}
	For all $\epsilon<\max\{c_1,c_2\}$, we have
	\begin{align}
	\eqref{problem}&\geq\!\!\!\!\!\!\!\!\!\! \max_{\begin{array}{c} 
			\substack{J_u,\ P_U: X-Y-U,\\ J_u\ \text{satisfies}\ \eqref{prop1},\ \eqref{prop2},\ \text{and}\ \eqref{jadid}}
	\end{array}} \!\!\!\!\!\!\!\!\!\!\!\!\!\!\!\!I(Y;U)\nonumber\\&\cong P_1\nonumber\\&\geq\frac{1}{2}\epsilon^2\frac{\sigma_{\text{max}}^2}{\gamma_1\gamma_2}\geq \frac{1}{2}\epsilon^2\frac{\sigma_{\text{max}}^2}{\gamma_{\max}^2},
\end{align}
	where $\gamma_1$, $\gamma_2$ and $\gamma_{\max}$ are defined in Lemma \ref{lem2}. Furthermore, $c_1$ and $c_2$ are defined in Proposition \ref{pos3}.
\end{theorem}
\begin{proof}
	The proof is based on Corollary \ref{cor1}, Proposition \ref{pos3} and Lemma \ref{lem2}.
\end{proof}
After finding $L_u$ and $P_U$ that attain the lower bounds in Theorem \ref{th1}, we can find the joint distribution $P_{XYU}$ as follows
\begin{align*}
P_{X|U=0}&=P_X+\epsilon[\sqrt{P_X}]L_1,\\
P_{X|U=1}&=P_X+\epsilon[\sqrt{P_X}]L_2,\\
P_{Y|U=0}&=P_X+\epsilon P_{X|Y}^{-1}[\sqrt{P_X}]L_1,\\
P_{Y|U=1}&=P_X+\epsilon P_{X|Y}^{-1}[\sqrt{P_X}]L_2.
\end{align*}
where $L_1$, $L_2$ and marginal distribution of $U$ are obtained in Lemma \ref{lem2}. Finally, $P_{XYU}(x,y,u)=P_{X|Y}(x|y)P_{Y|U}(y|u)P_U(u)$.
\subsection{Second Approach: Direct Approximation of \eqref{problem}}
In this section, we use \eqref{c3} which is equivalent to the bounded LIP in \eqref{local1} to approximate \eqref{problem}. Let us recall that the constraint in \eqref{jadid}, i.e., $\frac{-\sqrt{P_X}}{1+\epsilon}\leq L_u\leq \sqrt{P_X},\ \forall u$, implies \eqref{c3}, since we have $e^{\epsilon}\geq \epsilon+1$ and $e^{-\epsilon}\leq \frac{1}{1+\epsilon}$.
\begin{proposition}\label{sw}
	For all $\epsilon<\max\{c_1',c_2'\}$ and invertible $P_{X|Y}$, we have
	\begin{align*}
	\max_{\begin{array}{c} 
		\substack{P_{U|Y}: X-Y-U,\\ -\epsilon\leq \log(\frac{P_{X|U}(x|u)}{P_{X}(x)})\leq \epsilon, \forall x,u}
		\end{array}} \!\!\!\!\!\!\!\!\!\!\!\!\!\!\!\!I(Y;U)&=P_2+o(\epsilon^2)\cong P_2\geq P_1,
	\end{align*}
	where 
	\begin{align}\label{pp}
	P_2\triangleq\!\!\!\!\! \max_{\begin{array}{c} 
		\substack{L_u,P_U: L_u \text{and} P_U \text{satisfy}\\  \eqref{c1}, \eqref{c2}, \text{and}\ \eqref{c3}}
		\end{array}} \!\!\!\!\!\!\!\!0.5\epsilon^2\!\!\left(\!\sum_u \!\!P_U\|WL_u\|^2\! \right),
	\end{align}
	and $c_1' \triangleq \log\left(\frac{\min_y P_Y(y)}{\max_i \left(\sum_{j=1}^{\mathcal{K}}|\alpha_{ij}|P_X(x) \right)}+1\right)$, $c_2'\triangleq \log\left(|\sigma_{\text{min}}(P_{X|Y})|\left(\min_y P_Y(y)\right)+1\right)$. 
\end{proposition}
\begin{proof}
	The equality follows from \eqref{approx1} and the inequality holds since the privacy constraint $\frac{-\sqrt{P_X}}{1+\epsilon}\leq L_u\leq \sqrt{P_X},\ \forall u$, implies \eqref{c3}. To obtain the bounds on $\epsilon$, similar to Proposition \ref{pos3} we must have $|\epsilon\frac{P_{X|Y}^{-1}J_u(y)}{P_Y(y)}|<1$ for all $u$ and $y$. One sufficient condition for $\epsilon$ to satisfy this inequality is to have $\epsilon<\log\left(|\sigma_{\text{min}}(P_{X|Y})|\left(\min_y P_Y(y)\right)+1\right)=c_2'$, since in this case we have
	\begin{align*}
	\epsilon^2|P_{X|Y}^{-1}J_u(y)|^2&\leq\epsilon^2\left\lVert P_{X|Y}^{-1}J_u\right\rVert^2\leq\epsilon^2 \sigma_{\max}^2\left(P_{X|Y}^{-1}\right)\left\lVert J_u\right\rVert^2\\&\stackrel{(a)}\leq\frac{(e^{\epsilon}-1)^2}{\sigma^2_{\text{min}}(P_{X|Y})}<\min_{y\in\mathcal{Y}} P_Y^2(y),
	\end{align*}
	which implies $|\epsilon\frac{P_{X|Y}^{-1}J_u(y)}{P_Y(y)}\!|<1$. The step (a) follows from \eqref{c3}, since \eqref{c3} implies $\|J_u\|^2\leq (\frac{e^\epsilon-1}{\epsilon})\sum_x P_X(j)= \frac{e^\epsilon-1}{\epsilon}$. Furthermore, another sufficient condition for $\epsilon$ to satisfy the inequality $|\epsilon\frac{P_{X|Y}^{-1}J_u(y)}{P_Y(y)}|<1$ is to have $\epsilon<\log\left(\frac{\min_y P_Y(y)}{\max_i \left(\sum_{j=1}^{\mathcal{K}}|\alpha_{ij}|P_X(x) \right)}+1\right)=c_1'$, since in this case we have
	\begin{align*}
	\epsilon|P_{X|Y}^{-1}J_u(y)|&\stackrel{(a)}{\leq} \epsilon \sum_x |\alpha_{yx}||J_u(x)|\\&\stackrel{(b)}{\leq} (e^\epsilon-1) \sum_x |\alpha_{yx}|P_X(x)\\ &< \min_y P_Y(y) \frac{\sum_x |\alpha_{yx}|P_X(x)}{\max_y \left(\sum_{x=1}^{\mathcal{K}}|\alpha_{yx}|P_X(x) \right)}\\&\leq \min_y P_Y(y),
	\end{align*}
	where (a) follows by $P_{X|Y}^{-1}J_u(y)=\sum_x \alpha_{yx}J_u(x)$ and $|\sum_i a_ib_i|\leq \sum_i |a_i||b_i|$. Furthermore, (b) follows by the privacy constraint which implies $|J_u(x)|\leq \frac{e^\epsilon-1}{\epsilon}P_X(x)$.
\end{proof}
Next, we study $P_2$ and find lower and upper bounds on it. We show that the lower bounds are tight up to a constant factor. Similar to Lemma \ref{lem2}, let $\sigma_{\text{max}}$ and $L^*$ be the maximum singular value of matrix $W$ and the corresponding singular vector with $\|L^*\|=1$. 
\begin{lemma}\label{lem3}
	 Let $\lambda_1>0$ and $\lambda_2>0$ be the largest possible scaling factors such that $\lambda_1L^*$ and $-\lambda_2L^*$ satisfy \eqref{c3}, i.e., we scale $L^*$ by $\lambda_1$ and $-L^*$ by $\lambda_2$. Both $\lambda_1$ and $\lambda_2$ can be larger or smaller than 1. Furthermore, let $\lambda'$ be the largest factor such that $\lambda'L^*$ and $-\lambda'L^*$ satisfy \eqref{c3}. Then, we have
	\begin{align*}
	\frac{1}{2}\epsilon^2\sigma_{\text{max}}^2(\lambda')^2\leq\frac{1}{2}\epsilon^2\sigma_{\text{max}}^2\lambda_1\lambda_2\leq P_2&\leq
	 \frac{1}{2}\epsilon^2\sigma_{\text{max}}^2\left(\frac{e^{\epsilon}-1}{\epsilon}\right)^2\\&=\frac{1}{2}\sigma_{\text{max}}^2\left(e^{\epsilon}-1\right)^2.
	\end{align*}
	Finally, for $|\mathcal{X}|=|\mathcal{Y}|=\mathcal{K}=2$, we have
	\begin{align}
	P_2=\frac{1}{2}\epsilon^2\sigma_{\text{max}}^2\left(\lambda_1\lambda_2\right).
	\end{align}
\end{lemma}
\begin{proof}
	The proof is provided in Appendix A.
\end{proof}
\begin{remark}
	In Lemma \ref{lem3}, if $L^*$ satisfies \eqref{c3}, we scale up $L^*$, i.e., we scale $L^*$ by the largest possible $\lambda_1\geq1$. Otherwise, we scale it down by largest possible $\lambda_1\leq1$. Similarly, if $-L^*$ satisfies \eqref{c3}, we scale it by $\lambda_2\geq1$, otherwise we scale it down. 
\end{remark}
\begin{proposition}\label{prop5}
	Both $L^*$ and $-L^*$ with $\|L^*\|=1$ do not simultaneously satisfy \eqref{c3}.
\end{proposition}
\begin{proof}
	Assume that both $L^*$ and $-L^*$ satisfy \eqref{c3}. We have
	\begin{align}
	\left(\frac{e^{-\epsilon}-1}{\epsilon}\right)\sqrt{P_{X}}\leq L^* \leq \left(\frac{e^{\epsilon}-1}{\epsilon}\right)\sqrt{P_{X}},\label{s}\\
	\left(\frac{e^{-\epsilon}-1}{\epsilon}\right)\sqrt{P_{X}}\leq -L^* \leq \left(\frac{e^{\epsilon}-1}{\epsilon}\right)\sqrt{P_{X}},\label{ss}
	\end{align}
	Multiplying \eqref{s} by $-1$ and combining it with \eqref{ss} we obtain
	\begin{align}
	\left(\frac{e^{-\epsilon}-1}{\epsilon}\right)\sqrt{P_{X}}\leq -L^*\leq \left(\frac{1-e^{-\epsilon}}{\epsilon}\right)\sqrt{P_{X}},\label{sss}
	\end{align}  
	where the last line follows since for $1>\epsilon>0$, $\frac{e^{\epsilon}-1}{\epsilon}> 1> |\frac{e^{-\epsilon}-1}{\epsilon}|$. 
	Note that for $1>\epsilon>0$, $\frac{1-e^{-\epsilon}}{\epsilon}< 1$ since $e^{\epsilon}\geq \epsilon+1$ implies $e^{-\epsilon}\leq \frac{1}{1+\epsilon}\leq \frac{1}{1-\epsilon}$. 
	Finally, \eqref{sss} implies
	\begin{align*}
	\|L^*\|^2\leq \left(\frac{1-e^{-\epsilon}}{\epsilon}\right)^2< 1,
	\end{align*}
	which contradicts $\|L^*\|=1$.
\end{proof}
Using Proposition \ref{prop5} we conclude that $\lambda_1\geq1$ and $\lambda_2\geq1$ are not feasible. Intuitively, this holds due to the lower bound in the privacy constraint described in \eqref{c3}.
\begin{remark}
	In Lemma \ref{lem2} and Lemma \ref{lem3}, the scaling factors $\gamma_1$, $\gamma_2$, $\lambda_1,$ and $\lambda_2$ depend on the value of $\epsilon$.
	\end{remark}
\begin{remark}
	In contrast to \eqref{c3}, the upper bound in \eqref{jadid} does not depend on $\epsilon$. This results in simpler mechanism designs when using the first approach. On the other hand, the second approach yields larger utilities compared to the first approach, as $P_2\geq P_1$. 
	\end{remark}
 In the next result we summarize the bounds and approximations on \eqref{problem}.
\begin{theorem}
	For all $\epsilon<\max\{c_1',c_2'\}$, we have
	\begin{align}
	\eqref{problem}\cong P_2\geq P_1,
	\end{align}
	where
	\begin{align}
	P_2&\geq \frac{1}{2}\epsilon^2\sigma_{\text{max}}^2\max\{\frac{1}{\gamma_1\gamma_2},\lambda_1\lambda_2\}\\&\geq \frac{1}{2}\epsilon^2\sigma_{\text{max}}^2\max\{\frac{1}{\gamma_{\text{max}}^2},(\lambda')^2\},
	\end{align}
	where $c_1'$ and $c_2'$ are defined in Proposition \ref{sw}.
\end{theorem}
\begin{proof}
	The proof is based on Lemma \ref{lem2} and Lemma \ref{lem3}.
\end{proof}
\begin{remark}
	Clearly, $c_1'\leq c_1$ and $c_2'\leq c_2$. Hence, to compare the proposed approaches we consider $\epsilon<\max\{c_1',c_2'\}$. 
\end{remark}
\begin{remark}
	To use the lower bounds on $P_1$ and $P_2$ as lower bounds on \eqref{problem}, we need to find bounds on the error of approximations. We have obtained bounds on error of approximation in \cite{shah} using the strong $\ell_1$-privacy criterion and bounded mutual information as the leakage constraints. Following similar method, we can bound error of approximations in this paper. Let $\eqref{problem}=P_2+\text{error}_2(\epsilon)\geq P_1+\text{error}_1(\epsilon)$. Using the methods in \cite{shah}, we have $|\text{error}_2(\epsilon)|\leq C_2^{\epsilon}$ and $|\text{error}_1(\epsilon)|\leq C_1^{\epsilon}$. Then, we can strengthen Theorem 1 and Theorem 2 and obtain
	\begin{align}
	\eqref{problem}\geq \max\{P_2-|C_2^{\epsilon}|, P_1-|C_1^{\epsilon}|\}.
	\end{align} 
\end{remark}
\section{Discussion}
\subsection{Geometrical Interpretation}
In this part, we present geometrical studies of the proposed approaches in this paper. Considering both methods, the constraint $L_u\perp \sqrt{P_X}$ leads to a subspace where each element (point) is orthogonal to $\sqrt{P_X}$. Furthermore, the privacy constraints $\frac{-\sqrt{P_X}}{1+\epsilon}\leq L_u\leq \sqrt{P_X}$ or $\left(\frac{e^{-\epsilon}-1}{\epsilon}\right)\sqrt{P_{X}}\leq L_u \leq \left(\frac{e^{\epsilon}-1}{\epsilon}\right)\sqrt{P_{X}}$ lead to a subspace where the boundaries are the hyperplanes corresponding the boundaries within the inequalities. For simplicity let $\mathcal{K}=2$, i.e., a 2D problem. Considering the first approach, feasible sets corresponding \eqref{c1} and \eqref{jadid} are shown in Fig. \ref{geo1}. Subspace $S_1$ which is shown by a red line corresponds to the points which are orthogonal to $\sqrt{P_X}$. Subspace $S_2$ which is inside the pink rectangle describes the points satisfying the privacy constraint. Finally, $S_3$ (dark red line) shows the intersection of $S_1$ and $S_2$. The goal is to find vectors $\{L_u\}$ satisfying \eqref{c2} which maximizes $\sum_u \!\!P_U\|WL_u\|^2$. 
\begin{figure}[]
	\centering
	\includegraphics[width = 0.45\textwidth]{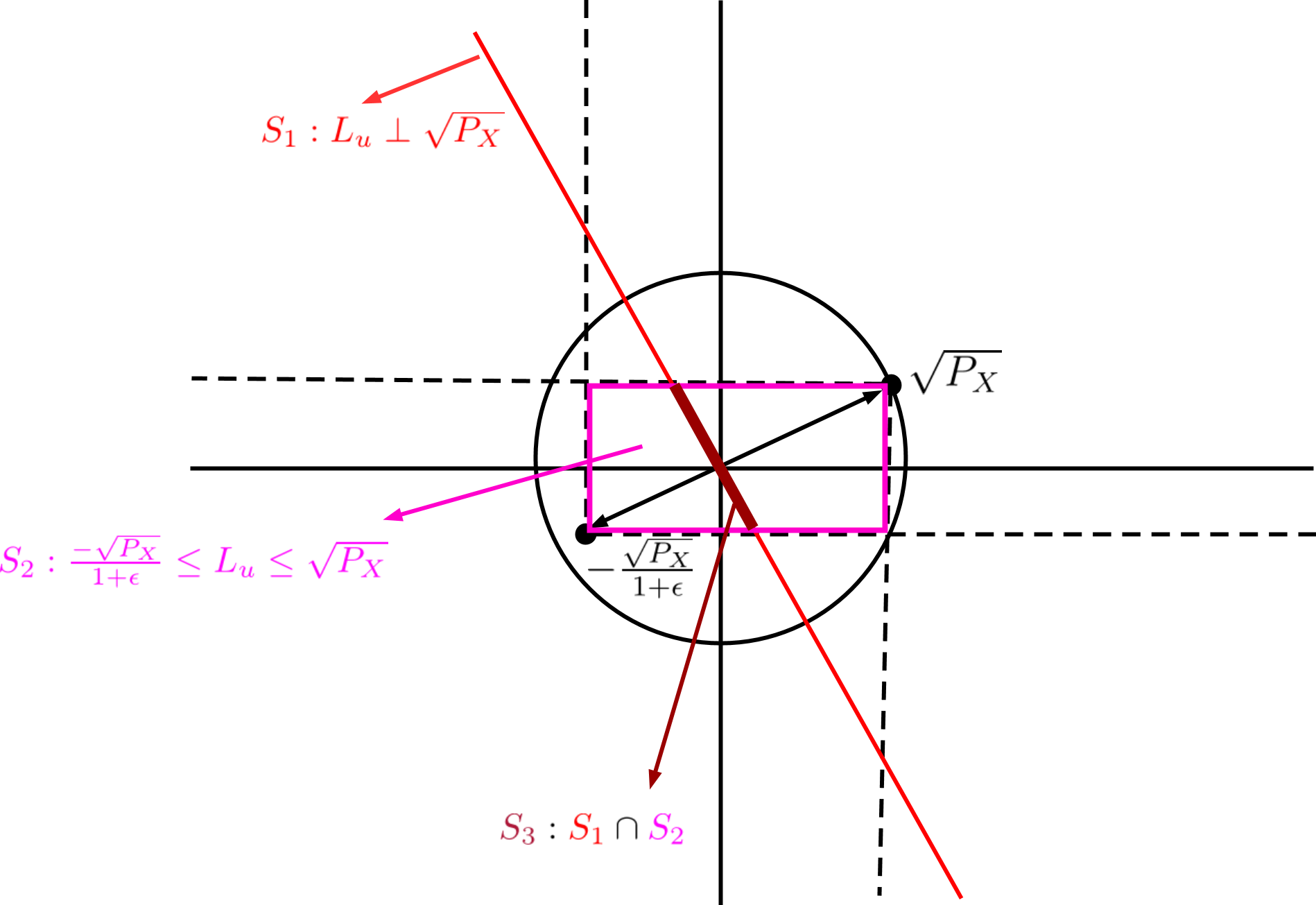}
	\caption{Feasible sets in the first approach. Here, $S_1$ and $S_2$ correspond to the points satisfying \eqref{c1} and \eqref{jadid}, respectively. The circle corresponds to the unit $\ell_2$-ball.}
	\label{geo1}
\end{figure}
Considering the second approach, feasible sets corresponding \eqref{c1} and \eqref{c3} are shown in Fig. \ref{geo2}. Here, subspace $S_2$ (inside the pink rectangle) corresponds to the points satisfying the privacy constraint in \eqref{c3}. As we can see, the feasible set $S_3$ is larger compared to the previous approach.
\begin{figure}[]
	\centering
	\includegraphics[width = 0.45\textwidth]{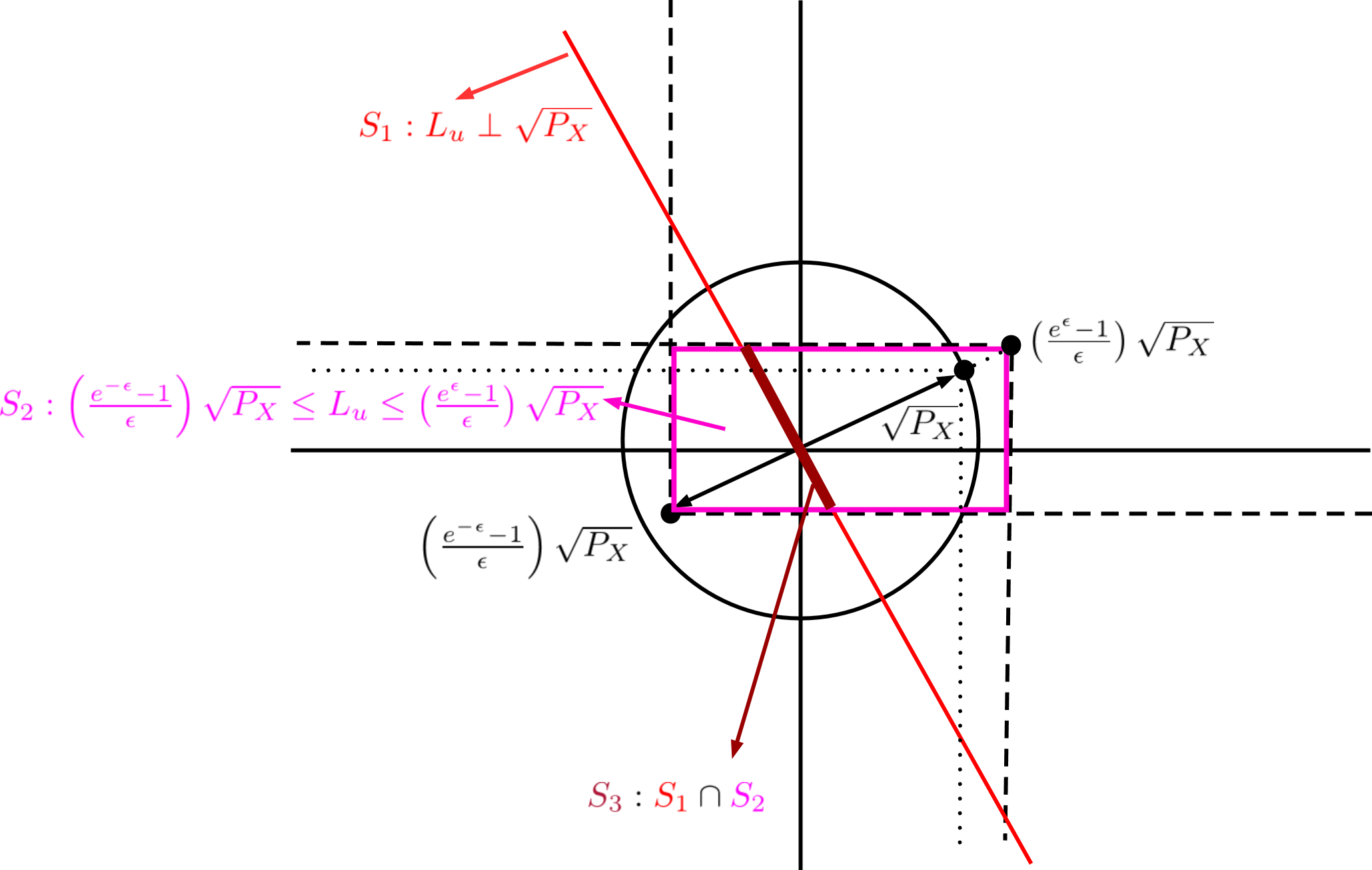}
	\caption{Feasible sets in the second approach. Here, $S_1$ and $S_2$ correspond to the points satisfying \eqref{c1} and \eqref{c3}, respectively. The circle corresponds to the unit $\ell_2$-ball. Compared to the first approach, the feasible set is larger which leads to higher utilities.}
	\label{geo2}
\end{figure}
In Fig. \ref{geo3}, long and short green arrows correspond to the vectors $\frac{L^*}{\gamma_1}$ and $-\frac{L^*}{\gamma_2}$ that achieve the lower bound in Lemma \ref{lem2}. Moreover, boundaries of the privacy constraint is shown by smaller dotted rectangle. We emphasize that we choose $\gamma_1$ and $\gamma_2$ so that both $\frac{L^*}{\gamma_1}$ and $-\frac{L^*}{\gamma_2}$ obtain largest possible $\ell_2$-norms (attain the boundaries). Furthermore, long and short dark red arrows correspond to the vectors $\lambda_1L^*$ and $-\lambda_2L^*$ that achieve the lower bound in Lemma \ref{lem3}. The boundaries of the privacy constraint is shown by larger dotted rectangle. Similarly, we choose $\lambda_1$ and $\lambda_2$ so that both $\lambda_1L^*$ and $-\lambda_2L^*$ obtain largest possible $\ell_2$-norms (attain the boundaries). Using Lemma \ref{lem2} and Lemma \ref{lem3}, $\frac{L^*}{\gamma_1}$ and $-\frac{L^*}{\gamma_2}$ as well as $\lambda_1L^*$ and $-\lambda_2L^*$, are optimal and attain $P_1$ and $P_2$, respectively. 
This discussion on the geometrical studies of the proposed approach to solve the privacy-utility trade-off in \eqref{problem} can easily be extended to any $\mathcal{K} > 2$. 
\begin{figure}[]
	\centering
	\includegraphics[width = 0.45\textwidth]{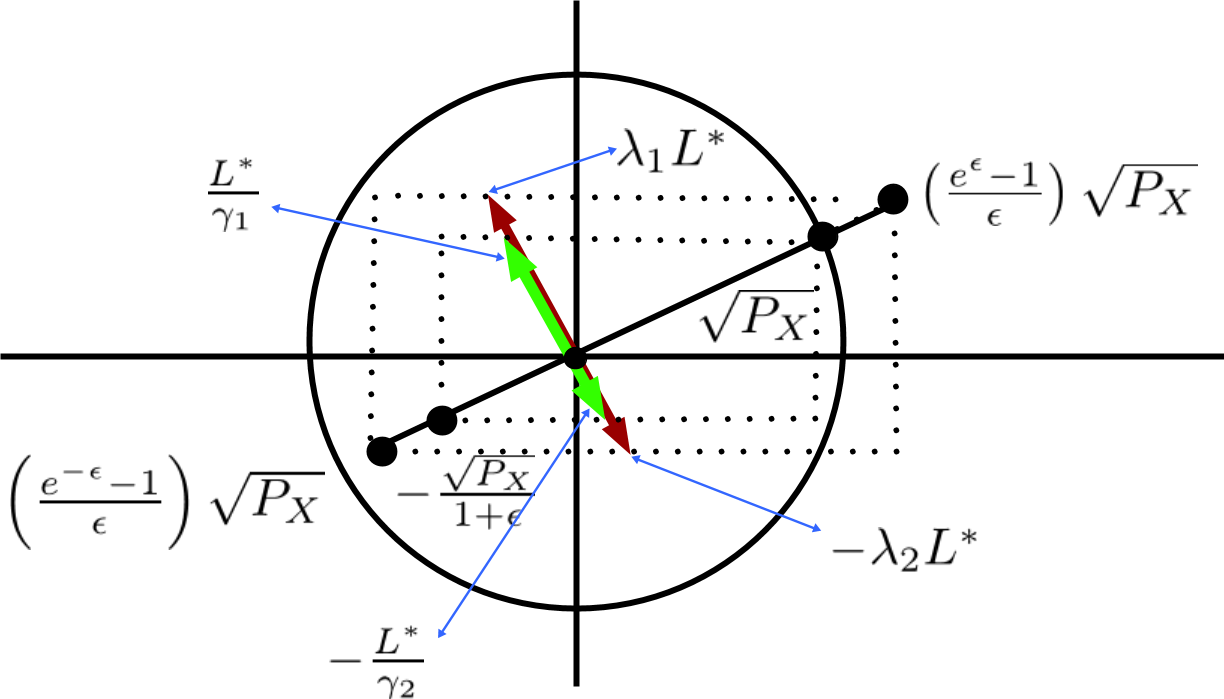}
	\caption{Acheivability of Lemma \ref{lem2} and Lemma \ref{lem3}. }
	\label{geo3}
\end{figure}
\subsection{Extending the approach considering max-lift leakage measure}
In this section, we extend the proposed approach to the privacy problem considered in \cite{zarab1}. Specifically, we use max-lift constrain instead of \eqref{local1}. Instead of \eqref{local1} we use
\begin{align}
\log(\frac{P_{X|U}(x|u)}{P_{X}(x)})\leq \epsilon, \forall x,u.\label{local2}
\end{align}
We have
\begin{subequations}\label{problem2}
	\begin{align}
&\sup_{P_{U|Y}} \ \ I(U;Y),\label{problem22}\\
&\text{subject to:}\ \ X-Y-U,\label{Markov2}\\
&\log(\frac{P_{X|U}(x|u)}{P_{X}(x)})\leq \epsilon, \forall x,u.\label{local22}
\end{align}
\end{subequations}
Clearly, the bounded LIP leads to the bounded max-lift criterion, i.e., \eqref{local1} results in \eqref{local2}. Hence, the maximization problem in \eqref{problem} is upper bounded by \eqref{problem} replacing \eqref{local1} with \eqref{local2}.
Following the first approach and using the upper bound $\log(1+x)\leq x$, we use the following strengthened privacy constraint instead of \eqref{local2} 
\begin{align}\label{pr}
L_u\leq \sqrt{P_X}, \forall u.
\end{align}
Let us define 
\begin{align}\label{lower2}
P_1'\triangleq\max_{\begin{array}{c} 
	\substack{L_u,P_U: \ L_u\leq \sqrt{P_X},\forall u ,\\ L_u\ \text{and}\ P_U\ \text{satisfy}\ \eqref{c1},\ \text{and}\ \eqref{c2}}
	\end{array}} \!\!\!\!\!\!\!\!\!\!\!\!\!\!\!\!\!\!\!0.5\epsilon^2\!\!\left(\sum_u \!\!P_U\|WL_u\|^2 \right).
\end{align}
Using similar techniques as Proposition \ref{pos3}, we have
\begin{align}
\eqref{problem2}\geq P_1'+o(\epsilon)\cong P_1'.
\end{align}
\begin{lemma}\label{22}
	If $L^*$ and $-L^*$ with $\|L^*\|=1$, satisfy \eqref{pr}, we have
	\begin{align}\label{ant}
	P_1'= \frac{1}{2}\epsilon^2\sigma_{\text{max}}^2.
	\end{align}
	Otherwise, we scale $L^*$ and $-L^*$ with smallest possible $\gamma_1\geq1$ and $\gamma_2\geq1$  so that $\frac{L^*}{\gamma_1}$ and $-\frac{L^*}{\gamma_2}$ satisfy \eqref{pr}. We have
	\begin{align}
	\frac{1}{2}\epsilon^2\frac{\sigma_{\text{max}}^2}{\gamma_1\gamma_2}\leq P_1'\leq \frac{1}{2}\epsilon^2\sigma_{\text{max}}^2.
	\end{align} 
	Finally, for $|\mathcal{X}|=|\mathcal{Y}|=\mathcal{K}=2$, we have
	\begin{align}
	P_1'=\frac{1}{2}\epsilon^2\frac{\sigma_{\text{max}}^2}{\gamma_1\gamma_2}.\label{css}
	\end{align}
\end{lemma}
\begin{proof}
	Here, we only prove \eqref{ant} and other statements can be shown by using similar proof as Lemma \ref{lem2}. If $L^*$ and $-L^*$ with $\|L^*\|=1$, satisfy \eqref{pr}, we choose $U$ to be a binary RV with weights $P_U^1=P_U^2=\frac{1}{2}$ and $L_1=L^*$, $L_2=-L^*$. We have
	\begin{align}
	\frac{1}{2}\epsilon^2\!\!\left(\sum_u \!\!P_U\|WL_u\|^2\right)=  \frac{1}{2}\epsilon^2\sigma_{\text{max}}^2.
	\end{align}
\end{proof}
Following the second approach, \eqref{local2} can be rewritten as 
\begin{align}\label{c33}
L_u \leq \left(\frac{e^{\epsilon}-1}{\epsilon}\right)\sqrt{P_{X}}, \ \forall x,u.
\end{align}
\begin{lemma}\label{33}
	If $L^*$ and $-L^*$ with $\|L^*\|=1$, satisfy \eqref{c33}, we scale them with largest possible $\lambda_1\geq 1$ and $\lambda_2\geq 1$ so that $\lambda_1L^*$ and $-\lambda_2L^*$ satisfy \eqref{c33}. Otherwise, we scale them with largest possible $\lambda_1> 0$ and $\lambda_2> 0$. We have
	\begin{align*}
	\frac{1}{2}\epsilon^2\sigma_{\text{max}}^2\lambda_1\lambda_2\leq P_2'\leq\frac{1}{2}\sigma_{\text{max}}^2\left(e^{\epsilon}-1\right)^2.
	\end{align*} 
	Furthermore, for $|\mathcal{X}|=|\mathcal{Y}|=\mathcal{K}=2$, we have
	\begin{align}
	P_2'=\frac{1}{2}\epsilon^2\sigma_{\text{max}}^2\left(\lambda_1\lambda_2\right).
	\end{align}
\end{lemma}
\begin{proof}
	The proof is similar to Lemma \ref{lem3}. The only difference is that both $L^*$ and $-L^*$ with $\|L^*\|=1$ can simultaneously satisfy \eqref{c33}. In other words, we can have $\lambda_1\geq 1$ and $\lambda_2\geq 1$ simultaneously.
\end{proof}
\begin{remark}
	In contrast to Proposition \ref{prop4} and Proposition \ref{prop5}, by using max-lift instead of LIP, both $L^*$ and $-L^*$ with $\|L^*\|=1$ can simultaneously satisfy \eqref{pr} or \eqref{c33}. This extends the optimality conditions of $P_1'$ and $P_2'$ compared to $P_1$ and $P_2$, i.e., the lower bounds on $P_1'$ and $P_2'$ are optimal for more cases compared to $P_1$ and $P_2$.      
\end{remark}
\begin{theorem}
	Let $\epsilon$ be sufficiently small. We have
	\begin{align}
	\eqref{problem2}\cong P_2'\geq P_1',
	\end{align}
	where
	\begin{align}
	P_2'&\geq \frac{1}{2}\epsilon^2\sigma_{\text{max}}^2\max\{\frac{1}{\gamma_1\gamma_2},\lambda_1\lambda_2\}.
	\end{align}
\end{theorem}
\begin{proof}
	The proof is based on Lemma \ref{22} and Lemma \ref{33}.
\end{proof}
\subsection{Extending the approach considering Local Differential Privacy (LDP)}
In this section, we discuss how the proposed approach can be applied for the privacy problem with LDP as the leakage constraint. Specifically, we use LDP instead of \eqref{local1}. We use
\begin{align}
\log\left(\frac{P_{U|X}(U=u|X=x)}{P_{U|X}(U=u|X=x')}\right)\leq \epsilon, \forall x,x'\in\mathcal{X},\ u\in \mathcal{U}.\label{local3}
\end{align}
We rewrite the left hand side of \eqref{local3} as follows
\begin{align}\label{dif1}
\log\left(\frac{P_{U|X}(U=u|X=x)}{P_{U|X}(U=u|X=x')}\right)= \log\left(\frac{\frac{P_{X|U}(X=x|U=u)}{P_{X}(x)}}{\frac{P_{X|U}(X=x'|U=u)}{P_{X}(x')}}\right)
\end{align}
Considering high privacy regimes (noisy databases), i.e., small $\epsilon$, using LDP we can say that $U$ and $X$ are approximately independent. 
Hence, we can introduce a perturbation vector as 
\begin{align}\label{dif2}
P_{X|U=u}=P_X+\epsilon J_u.
\end{align} 
Using \eqref{dif1} and \eqref{dif2}, we have
\begin{align}\label{to}
\log\left(\frac{P_{U|X}(U=u|X=x)}{P_{U|X}(U=u|X=x')}\right)= \log\left(\frac{1+\epsilon\frac{J_u(x)}{P_X(x)}}{1+\epsilon\frac{J_u(x')}{P_X(x')}}\right)
\end{align}
To follows the first approach using \eqref{to} we have
\begin{align}\label{tooo}
\log(1+\epsilon\frac{J_u(x)}{P_X(x)})-\log(1+\epsilon\frac{J_u(x')}{P_X(x')})\leq \epsilon.
\end{align}
We then use upper and lower bounds on $\log(1+x)$ so that we get strengthened privacy constraints as follows
\begin{align}\label{too}
\frac{J_u(x)}{P_X(x)}\leq \frac{J_u(x')}{P_X(x')+\epsilon J_u(x')}, \forall x,x',u.
\end{align}
We recall that \eqref{too} implies \eqref{tooo}.
Using the second approach the LDP constraint in \eqref{local3} can be rewritten as 
\begin{align}\label{mm}
\frac{J_u(x)}{P_X(x)}-e^\epsilon\frac{J_u(x')}{P_X(x')}\leq \frac{e^\epsilon-1}{\epsilon}, \forall x,x',u.
\end{align}
Finally, we approximate $I(U;Y)$ and solve the optimizations under the privacy constraints \eqref{too} or \eqref{mm} similar to Lemma \ref{lem2} and Lemma \eqref{lem3}. We leave the remainder of this discussion for an extended journal version.
\subsection{Higher Order Approximation}\label{dis3}
In this section, we discuss how our method can be applied to larger permissible leakage intervals. In this paper, to approximate $I(U;Y)$ we used second order Taylor expansion of $\log(1+x)$ that equals to $x-\frac{x^2}{2}+o(x^2)$. One way to increase the permissible leakage interval is to use higher order approximation of $\log(1+x)$. This follows since if we use approximations with higher orders Taylor expansions, to achieve the same error (of approximation) we can use higher values of $x$. For instance, let $\log(1+x)=x-\frac{x^2}{2}+\text{error}_1(x)$ and $\log(1+x)=x-\frac{x^2}{2}+\frac{x^3}{3}+\text{error}_2(x)$. If for the first approximation we use the interval $x\leq c_1$, then we can use the interval $x\leq c_2$ where $c_2>c_1$, for the second approximation to have the same error, i.e., $|\text{error}_1(x)|=|\text{error}_2(x)|$. Furthermore, we obtain tighter bounds compared to the current results (closer to the optimal value); however, the optimization problems become more complex. In order to approximate $I(U;Y)$ we use third order Taylor expansion and we get
\begin{align*}
I(Y;U)&=\sum_u P_U(u)D(P_{Y|U=u}||P_Y)\\&=\sum_u P_U(u)\sum_y P_{Y|U=u}(y)\log\left(\frac{P_{Y|U=u}(y)}{P_Y(y)}\right)\\&=\sum_u P_U(u)\sum_y\! P_{Y|U=u}(y)\log\left(\!1\!+\!\epsilon\frac{P_{X|Y}^{-1}J_u(y)}{P_Y(y)}\right)\\&\stackrel{(a)}{=}\frac{1}{2}\epsilon^2\sum_u P_U\\&\times \left(\sum_y
\frac{(P_{X|Y}^{-1}J_u)^2}{P_Y}  -\frac{1}{6}\epsilon^3\frac{(P_{X|Y}^{-1}J_u)^3}{P_Y^2}\right)+o(\epsilon^3)\\
&=o(\epsilon^3)+\frac{1}{2}\epsilon^2\sum_u P_U\\&\times\left(\!\|[\sqrt{P_Y}^{-1}]P_{X|Y}^{-1}J_u\|^2\!-\!\frac{1}{3}\epsilon\|[\sqrt{P_Y}^{-1}]P_{X|Y}^{-1}J_u\|_3^3\!\right)\\&=o(\epsilon^3)+\sum_u P_U\left(\frac{1}{2}\epsilon^2\|WL_u\|^2-\frac{1}{6}\epsilon^3\|WL_u\|_3^2\right)
\\&\cong\sum_u P_U\left(\frac{1}{2}\epsilon^2\|WL_u\|^2-\frac{1}{6}\epsilon^3\|WL_u\|_3^2\right),
\end{align*}
where (a) follows by  
\begin{align*}
\log(1+\epsilon\frac{P_{X|Y}^{-1}J_u}{P_y})=&\epsilon\frac{P_{X|Y}^{-1}J_u}{P_y}-\frac{1}{2}\epsilon^2(\frac{P_{X|Y}^{-1}J_u}{P_y})^2\\&+\frac{1}{3}\epsilon^3(\frac{P_{X|Y}^{-1}J_u}{P_y})^3+o(\epsilon^3).
\end{align*}
and $\|\cdot\|_3$ corresponds to $\ell_3$-norm. Finally, we use $\sum_u P_U\left(\frac{1}{2}\epsilon^2\|WL_u\|^2-\frac{1}{6}\epsilon^3\|WL_u\|_3^2\right)$ in \eqref{lower} and \eqref{pp} instead of $\frac{1}{2}\epsilon^2\!\!\left(\sum_u \!\!P_U\|WL_u\|^2 \right)$
\subsection{Extension to a general leakage matrix $P_{X|Y}$}\label{dis2}
 In this part, we discuss how the proposed approach can be extended for any matrix $P_{X|Y}$. To do this, we can follow a similar method as in \cite{Khodam22}. Let $P_{X|Y}$ be a full row rank matrix with $|\mathcal{X}|<|\mathcal{Y}|$, and without loss of generality we assume that $P_{X|Y}$ can be represented by two submatrices where the first submatrix is invertible, i.e., $P_{X|Y}=[P_{X|Y_1} , P_{X|Y_2}]$ such that $P_{X|Y_1}$ defined on $\mathbb{R}^{|\mathcal{X}|\times|\mathcal{X}|}$ is invertible. Furthermore, let $M\in \mathbb{R}^{|\mathcal{X}|\times|\mathcal{Y}|}$ be constructed as follows:
 Let $V$ be the matrix of right eigenvectors of $P_{X|Y}$, i.e., $P_{X|Y}=U\Sigma V^T$ and $V=[v_1,\ v_2,\ ... ,\ v_{|\mathcal{Y}|}]$, then $M$ is defined as
 \begin{align*}
 M \triangleq \left[v_1,\ v_2,\ ... ,\ v_{|\mathcal{X}|}\right]^T.  
 \end{align*}  
 In \cite[Lemma 1]{Khodam22}, we presented two properties of $M$. Using those properties we get the next result.
 \begin{lemma}
 	 Let the Markov chain $X-Y-U$ holds and $J_u$ satisfies the three properties \eqref{prop1}, \eqref{prop2} and \eqref{prop3} or \eqref{jadid}. For sufficiently small $\epsilon>0$, for every $u\in\mathcal{U}$, the vector $P_{Y|U=u}$ belongs to the following convex polytope $\mathbb{S}_{u}$
 	\begin{align*}
 	\mathbb{S}_{u} = \left\{y\in\mathbb{R}^{|\mathcal{Y}|}|My=MP_Y+\epsilon M\begin{bmatrix}
 	P_{X|Y_1}^{-1}J_u\\0
 	\end{bmatrix},\ y\geq 0\right\},
 	\end{align*}
 	where $\begin{bmatrix}
 	P_{X|Y_1}^{-1}J_u\\0
 	\end{bmatrix}\in\mathbb{R}^{|\cal Y|}$ and $J_u$ satisfies \eqref{prop1}, \eqref{prop2}, and \eqref{prop3} or \eqref{jadid}.
 \end{lemma}
\begin{proof}
	The proof is similar to \cite[Lemma 2]{Khodam22}.
\end{proof}
Then, using the previous Lemma we have the following equivalency.
\begin{theorem}
	We have the following equivalency 
	\begin{align}\label{equi}
	\min_{\begin{array}{c} 
		\substack{P_{U|Y}:X-Y-U\\ -\epsilon\leq \log(\frac{P_{X|U}(x|u)}{P_{X}(x)})\leq \epsilon,\ \forall x,u}
		\end{array}}\! \! \! \!\!\!\!\!\!\!\!\!\!\!\!\!\!\!\!H(Y|U) =\!\!\!\!\!\!\!\!\! \min_{\begin{array}{c} 
		\substack{P_U,\ P_{Y|U=u}\in\mathbb{S}_u,\ \forall u\in\mathcal{U},\\ \sum_u P_U(u)P_{Y|U=u}=P_Y,\\ J_u \text{satisfies}\ \eqref{prop1},\ \eqref{prop2},\ \text{and}\ \eqref{prop3}}
		\end{array}} \!\!\!\!\!\!\!\!\!\!\!\!\!\!\!\!\!\!\!H(Y|U),
	\end{align}
	where $P_U$ defined on $\mathbb{R}^{|\mathcal{U}|}$ is the marginal distribution of $U$.
\end{theorem}
Next, it can be shown that how $H(Y|U)$ is minimized over $P_{Y|U=u}\in\mathbb{S}_u$ for all $u\in\mathcal{U}$.
\begin{proposition}\label{4}
	Let $P^*_{Y|U=u},\ \forall u\in\mathcal{U}$ be the minimizer of $H(Y|U)$ over the set $\{P_{Y|U=u}\in\mathbb{S}_u,\ \forall u\in\mathcal{U}|\sum_u P_U(u)P_{Y|U=u}=P_Y\}$, then 
	$P^*_{Y|U=u}\in\mathbb{S}_u$ for all $u\in\mathcal{U}$ must belong to extreme points of $\mathbb{S}_u$.
\end{proposition}
Following the same approach in \cite{Khodam22} we can find the extreme points and approximate the right hand side in \eqref{equi}. We leave the remainder of this discussion for an extended journal version.
\section{Numerical Example}
In this section, we present numerical examples to evaluate the proposed approaches and compare the results with \cite{khodam}, where have used the \emph{strong $\chi^2$-privacy criterion} as the privacy constraint. To do so, we use \cite[Example 1]{khodam}.
\begin{example}\cite[Example 1]{khodam}
	Let the leakage matrix be $P_{X|Y}=\begin{bmatrix}
	\frac{1}{4}       & \frac{2}{5}  \\
	\frac{3}{4}    & \frac{3}{5}
	\end{bmatrix}$
	and $P_Y$ be given as $[\frac{1}{4} , \frac{3}{4}]^T$. Thus, we find $W$ and $P_X$ as
	\begin{align*}
	P_X&=P_{X|Y}P_Y=[0.3625, 0.6375]^T,\\
	W &= [\sqrt{P_Y}^{-1}]P_{X|Y}^{-1}[\sqrt{P_X}] = \begin{bmatrix}
	-4.8166       & 4.2583  \\
	3.4761    & -1.5366
	\end{bmatrix}.
	\end{align*}
	The singular values of $W$ are $7.4012$ and $1$ with corresponding right singular vectors $[0.7984, -0.6021]^T$ and $[0.6021 , 0.7984]^T$, respectively. Here, $L^*$ equals to $[0.7984, -0.6021]^T$ and we assume that $\epsilon<0.05$.
	Following the first approach, using the optimal vectors $\frac{L^*}{\gamma_1}$ and $-\frac{L^*}{\gamma_2}$, and the privacy constraint, we have
	\begin{align*}
	-\frac{0.6021}{1+\epsilon}\leq \frac{0.7984}{\gamma_1}\leq 0.6021,\\
	-\frac{7984}{1+\epsilon}\leq -\frac{0.6021}{\gamma_1}\leq 0.7984,\\
	-\frac{0.6021}{1+\epsilon}\leq -\frac{0.7984}{\gamma_2}\leq 0.6021,\\
	-\frac{7984}{1+\epsilon}\leq \frac{0.6021}{\gamma_2}\leq 0.7984,
	\end{align*}
	resulting
	\begin{align*}
	\gamma_1&\geq \max\{1.326, \frac{1+\epsilon}{1.326}\}=1.326,\\
	\gamma_2&\geq \max\{1.326(1+\epsilon),\frac{1}{1.326}\}=1.326(1+\epsilon).
	\end{align*}
	Hence, we choose $\gamma_1=1.326$ and $\gamma_2=1.326(1+\epsilon)$. Using Lemma \eqref{lem2}, we have
	\begin{align}
	P_1=\frac{1}{2}\epsilon^2(7.4012)^2\frac{1}{1.326^2(1+\epsilon)}=15.5771\frac{\epsilon^2}{1+\epsilon}.
	\end{align}
	Following the second approach, using the optimal vectors $\lambda_1L^*$ and $-\lambda_2L^*$, and the privacy constraint, we have
	\begin{align*}
	0.6021\frac{e^{-\epsilon}-1}{\epsilon}\leq 0.7984\lambda_1\leq 0.6021\frac{e^\epsilon-1}{\epsilon},\\
	0.7984\frac{e^{-\epsilon}-1}{\epsilon}\leq -0.6021\lambda_1\leq 0.7984\frac{e^\epsilon-1}{\epsilon},\\
0.6021\frac{e^{-\epsilon}-1}{\epsilon}\leq -0.7984\lambda_2\leq 0.6021\frac{e^\epsilon-1}{\epsilon},\\
	0.7984\frac{e^{-\epsilon}-1}{\epsilon}\leq 0.6021\lambda_2\leq 0.7984\frac{e^\epsilon-1}{\epsilon},
	\end{align*}
	leading to
	\begin{align*}
	\lambda_1&\leq \min\{1.326\frac{1-e^{-\epsilon}}{\epsilon},\frac{1}{1.326}\frac{e^\epsilon-1}{\epsilon}\}\\&=\frac{1}{1.326}\frac{e^\epsilon-1}{\epsilon}\\
	\lambda_2&\leq\min\{1.326\frac{e^{\epsilon}-1}{\epsilon},\frac{1}{1.326}\frac{1-e^{-\epsilon}}{\epsilon}\}\\&=\frac{1}{1.326}\frac{1-e^{-\epsilon}}{\epsilon}.
\end{align*}
We choose $\lambda_1=\frac{1}{1.326}\frac{e^\epsilon-1}{\epsilon}$ and $\lambda_2=\frac{1}{1.326}\frac{1-e^{-\epsilon}}{\epsilon}$. Using Lemma \ref{lem3}, we obtain
\begin{align}
P_2 &= \frac{1}{2}\epsilon^2(7.4012)^2\frac{(e^\epsilon-1)(1-e^{-\epsilon})}{(1.326)^2\epsilon^2}\\&=15.5771(e^{\epsilon}+e^{-\epsilon}-2).
\end{align}
By replacing the bounded LIP with the strong $\chi^2$-privacy criterion and using the method in \cite{khodam}, we find the approximate maximum utility as $\frac{1}{2}\epsilon^2(7.4012)^2=27.39\cdot \epsilon^2$.
In Fig. \ref{geo4}, we compare $P_1$, $P_2$, the approximate value using the strong $\chi^2$-privacy criterion instead of \eqref{local1} and the optimal solution of \eqref{problem} using exhaustive search. We recall that $P_1$ is the approximate of lower bound on \eqref{problem} and $P_2$ is the approximate value of \eqref{problem}. We can see that $P_2$ dominates $P_1$, and the gap between exact value of \eqref{problem} which is found by exhaustive search and $P_2$ is small in this privacy regime. Finally, the blue curve shows the approximate solution of \eqref{problem}, replacing \eqref{local1} by the strong $\chi^2$-criterion \cite{khodam}. Intuitively, the blue curve dominates \eqref{problem}, since the bounded LIP must hold for all $x$ and $u$, however, the strong $\chi^2$-criterion must hold for each $u$. In other words, bounded LIP is a point-wise measure with respect to $x$ and $u$, and the strong $\chi^2$-criterion is point-wise with respect to $u$, which means that the bounded LIP is a stronger measure compared to the strong $\chi^2$-criterion. Hence, approximate utility of \cite{khodam} is greater than the approximate value of \eqref{problem}. Finally, the mapping between $U$ and $Y$ can be calculated as follows. For the first approach we have
\begin{align*}
P_{Y|U=0}&=P_Y+\epsilon P_{X|Y}^{-1}[\sqrt{P_X}]\frac{L^*}{\gamma_1}\\
&=[0.25-2.4169\cdot\epsilon , 0.75+2.4169\cdot\epsilon]^T,\\
P_{Y|U=1}&=P_Y-\epsilon P_{X|Y}^{-1}[\sqrt{P_X}]\frac{L^*}{\gamma_2}\\
&=[0.25+2.4169\cdot\frac{\epsilon}{1+\epsilon} , 0.75-2.4169\cdot\frac{\epsilon}{1+\epsilon}]^T.
\end{align*}    
And following the second approach we obtain
\begin{align*}
P_{Y|U=0}&=P_Y+\epsilon P_{X|Y}^{-1}[\sqrt{P_X}]\lambda_1L^*\\
&\!\!\!\!\!\!\!\!\!\!\!\!=[0.25-2.4169\cdot(e^\epsilon-1) , 0.75+2.4169\cdot(e^\epsilon-1)]^T,\\
P_{Y|U=1}&=P_Y-\epsilon P_{X|Y}^{-1}[\sqrt{P_X}]\lambda_2L^*\\
&\!\!\!\!\!\!\!\!\!\!\!\!=[0.25+2.4169\cdot(1-e^\epsilon) , 0.75-2.4169\cdot(1-e^\epsilon)]^T.
\end{align*}
For both approaches we can verify that $P_U^1P_{Y|U=0}+P_U^2P_{Y|U=1}=P_Y$.
\begin{figure}[]
	\centering
	\includegraphics[width = 0.48\textwidth]{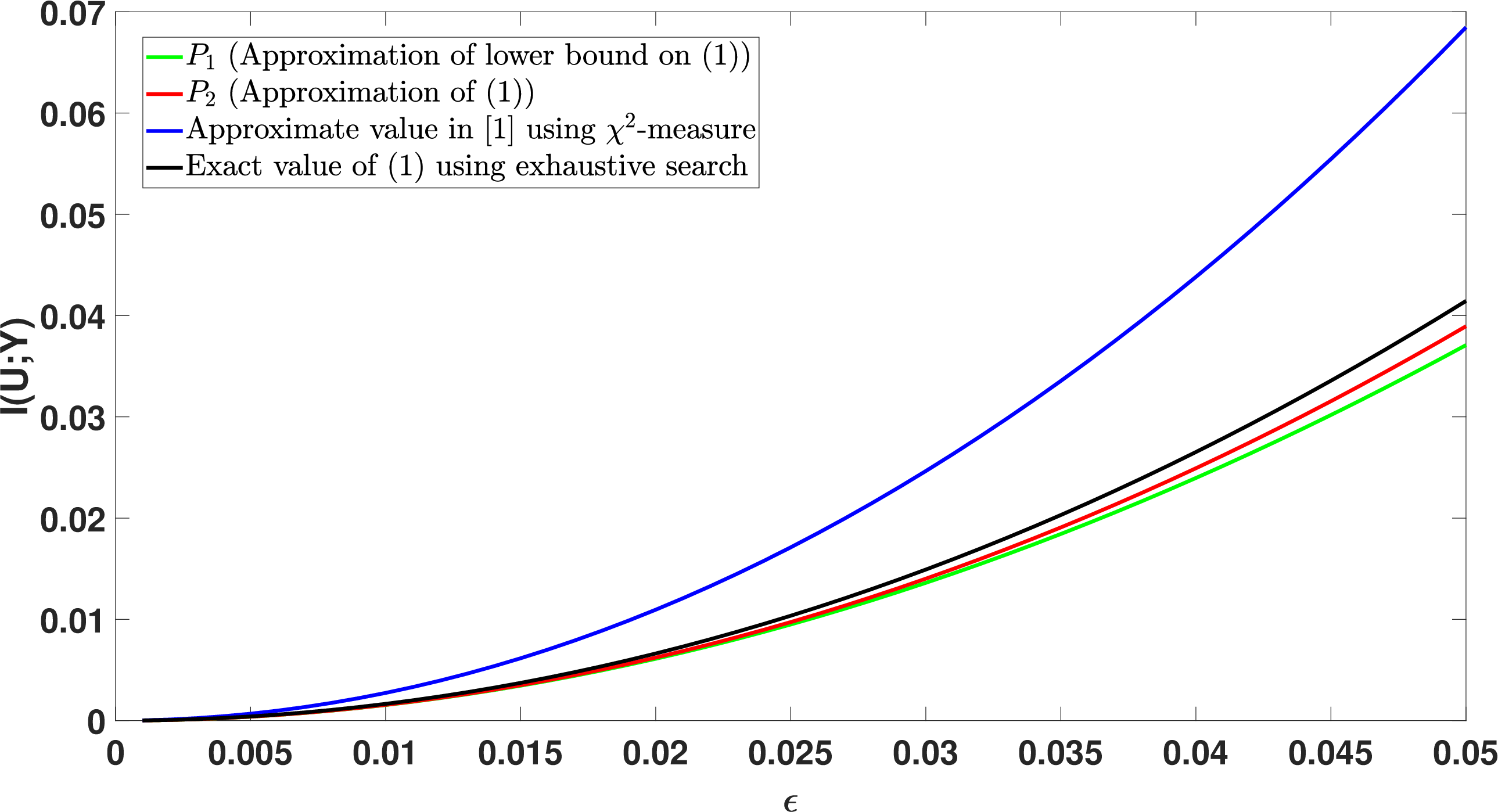}
	\caption{Comparing the proposed methods in this paper with the optimal solution and previous method in \cite{khodam}. It can be seen that in the high privacy regimes, $P_2$ is close to the exact solution which is found by exhaustive search and $P_2$ dominates $P_1$.}
	\label{geo4}
\end{figure}
\end{example}
\section{conclusion}\label{concul}
We have shown that information geometry can be used to simplify an information-theoretic privacy mechanism design problem with bounded LIP as the leakage constraint. When a small $\epsilon$ privacy leakage is allowed, simple approximate solutions are derived. A geometrical interpretation of the privacy mechanism design is provided. Specifically, we look for vectors satisfying the privacy constraints of having the largest Euclidean norm, leading to finding the largest principle singular value and vector of a matrix. The proposed approach establishes a useful and general design framework, which has been demonstrated in other privacy design problems such as considering maxlift or LDP instead of LIP. 
\section*{Appendix A}
\subsection*{Proof of Lemma \ref{lem2}:}
To derive the upper bound we have
\begin{align*}
\frac{1}{2}\epsilon^2\!\!\left(\sum_u \!\!P_U\|WL_u\|^2\right)&\leq \frac{1}{2}\epsilon^2 \left(\sum_u \!\!P_U\sigma_{\text{max}}^2\|L_u\|^2\right)\\&\stackrel{(a)}{\leq} \frac{1}{2}\epsilon^2\sigma_{\text{max}}^2\left(\sum_u \!\!P_U\right)=\frac{1}{2}\epsilon^2\sigma_{\text{max}}^2,
\end{align*}
where (a) follows by the privacy constraint, i.e., we have
\begin{align*}
\frac{-\sqrt{P_X}}{1+\epsilon}\leq L\leq \sqrt{P_X} \Rightarrow \|L\|^2\leq \max\{1,\frac{1}{(1+\epsilon)^2}\}=1.
\end{align*}
To derive the first lower bound, let us divide $L^*$ and $-L^*$ by $\gamma_1$ and $\gamma_2$ so that the privacy constraint is met. Let $U$ be a binary RV with probability masses $P_U^1=\frac{\gamma_1}{\gamma_1+\gamma_2}$ and $P_U^2=\frac{\gamma_2}{\gamma_1+\gamma_2}$. Furthermore, let $L_1=\frac{L^*}{\gamma_1}$ and $L_2=\frac{-L^*}{\gamma_2}$. Clearly, $L_1\perp \sqrt{P_X}$ and $L_2\perp \sqrt{P_X}$ and $P_U^1L_1+P_U^2L^2=0$. Moreover, we have
\begin{align*}
&\frac{1}{2}\epsilon^2\!\!\left(\sum_u \!\!P_U\|WL_u\|^2\right)=\\&\frac{1}{2}\epsilon^2\sigma_{\text{max}}^2\left( \frac{1}{\gamma_1(\gamma_1+\gamma_2)}+ \frac{1}{\gamma_2(\gamma_1+\gamma_2)}\right)=\frac{1}{2}\epsilon^2\frac{\sigma_{\text{max}}^2}{\gamma_1\gamma_2}.
\end{align*}
To derive the last lower bound let us divide $L^*$ and $-L^*$ by $\gamma_{\text{max}}$. In this case, since we divide $L^*$ and $-L^*$ by same factor we have $\gamma_{\text{max}}\geq \max\{\gamma_1,\gamma_2\}$.  Then, let $U$ be a binary RV with uniform distribution and $L_1=\frac{L^*}{\gamma_{\text{max}}}$ and $L_2=\frac{-L^*}{\gamma_{\text{max}}}$. Clearly, \eqref{c2} and \eqref{c3} are satisfied and 
\begin{align*}
\frac{1}{2}\epsilon^2\!\!\left(\sum_u \!\!P_U\|WL_u\|^2\right)=\frac{1}{2}\epsilon^2\frac{\sigma_{\text{max}}^2}{\gamma_{\max}^2}.
\end{align*} 
By using $\gamma_{\text{max}}\geq \max\{\gamma_1,\gamma_2\}$ we obtain the second lower bound.
Finally, to prove \eqref{cs}, we note that the only feasible directions are $L^*$ and $-L^*$. The latter follows since the singular vectors of $P_{X|Y}$ are $L^*$ and $\sqrt{P_X}$, which span the 2-D space. Hence, it is sufficient to consider binary $U$ and we have
\begin{align}\label{t}
\frac{1}{2}\epsilon^2\!\!\left(\sum_u \!\!P_U\|WL_u\|^2\right) \leq  \frac{1}{2}\epsilon^2\sigma_{\text{max}}^2\left(\frac{P_{U}^1}{\gamma_1^2}+\frac{P_{U}^2}{\gamma_2^2}\right)
\end{align}
where, $\gamma_1\geq1$ and $\gamma_2\geq1$ are the smallest possible factors ensuring that $L^*$ and $-L^*$ satisfy the privacy constraint. Furthermore, using the constraint 
\begin{align}
P_{U}^1L_1+P_{U}^2L_2=P_{U}^1\frac{L^*}{\gamma_1}-P_{U}^2\frac{L^*}{\gamma_2}=0
\end{align} 
and $P_{U}^1+P_{U}^2=1$,
we obtain
\begin{align}\label{tt}
P_{U}^1=\frac{\gamma_1}{\gamma_1+\gamma_2},\ P_{U}^2=\frac{\gamma_2}{\gamma_1+\gamma_2}. 
\end{align}
Combining \eqref{t} with \eqref{tt}, we have
\begin{align}\label{ttt}
\frac{1}{2}\epsilon^2\!\!\left(\sum_u \!\!P_U\|WL_u\|^2\right) \leq  \frac{1}{2}\epsilon^2\frac{\sigma_{\text{max}}^2}{\gamma_1\gamma_2}.
\end{align}
To attain the upper bound \eqref{ttt}, we let $P_{U}^1=\frac{\gamma_1}{\gamma_1+\gamma_2},\ P_{U}^2=\frac{\gamma_2}{\gamma_1+\gamma_2}$ and $L_1=\frac{L^*}{\gamma_1}$, $L_1=-\frac{L^*}{\gamma_2}$. We emphasize that the utilities in this Lemma are optimized over $\gamma_1\geq1$ and $\gamma_2\geq1$.
\subsection*{Proof of Lemma \ref{lem3}:}
To obtain the upper bound we have
\begin{align*}
\frac{1}{2}\epsilon^2\!\!\left(\sum_u \!\!P_U\|WL_u\|^2\right)&\leq \frac{1}{2}\epsilon^2 \left(\sum_u \!\!P_U\sigma_{\text{max}}^2\|L_u\|^2\right)\\&\stackrel{(a)}{\leq} \frac{1}{2}\epsilon^2\sigma_{\text{max}}^2\left(\sum_u \!\!\frac{e^{\epsilon}-1}{\epsilon}\right)^2\\&=\frac{1}{2}\sigma_{\text{max}}^2\left(e^{\epsilon}-1\right)^2,
\end{align*}
where (a) follows by the privacy constraint, i.e., \eqref{c3} implies
\begin{align*}
\|L\|^2&\leq \max\{\left(\frac{e^{-\epsilon}-1}{\epsilon}\right)^2,\left(\frac{e^{\epsilon}-1}{\epsilon}\right)^2\}\\&=\left(\frac{e^{\epsilon}-1}{\epsilon}\right)^2.
\end{align*}
To derive the lower bounds, let us scale $L^*$ and $-L^*$ by $\lambda_1$ and $\lambda_2$ so that the privacy constraint in \eqref{c3} is met. Let $U$ be a binary RV with weights $P_U^1=\frac{\lambda_2}{\lambda_1+\lambda_2}$ and $P_U^2=\frac{\lambda_1}{\lambda_1+\lambda_2}$. Furthermore, let $L_1=\lambda_1L^*$ and $L_2=-\lambda_1L^*$. Clearly, $L_1\perp \sqrt{P_X}$ and $L_2\perp \sqrt{P_X}$ and $P_U^1L_1+P_U^2L^2=0$. Moreover, we have
\begin{align*}
\frac{1}{2}\epsilon^2\!\!\left(\sum_u \!\!P_U\|WL_u\|^2\right)&=\frac{1}{2}\epsilon^2\sigma_{\text{max}}^2\left( \frac{\lambda_2\lambda_1^2+\lambda_1\lambda_2^2}{\lambda_1+\lambda_2}\right)\\&= \frac{1}{2}\epsilon^2\sigma_{\text{max}}^2\lambda_1\lambda_2\geq \frac{1}{2}\epsilon^2\sigma_{\text{max}}^2(\lambda')^2,
\end{align*}
where the last line follows by $\lambda'\leq \min\{\lambda_1,\lambda_2\}$. To attain the last inequality we let $P_U^1=P_U^2=\frac{1}{2}$ and $L_1=\lambda'L^*$, $L_2=-\lambda'L^*$. Similar to Lemma \ref{lem2}, when $\mathcal{K}=2$ the only feasible directions are $L^*$ and $-L^*$. Using the constraint $P_{U}^1L_1+P_{U}^2L_2=0$ with $L_1=\lambda_1L^*$ and $L_2=-\lambda_1L^*$, we obtain
\begin{align}\label{tth}
P_{U}^1=\frac{\lambda_2}{\lambda_1+\lambda_2},\ P_{U}^2=\frac{\lambda_1}{\lambda_1+\lambda_2}. 
\end{align} 
Thus, we have
\begin{align}
P_2=\frac{1}{2}\epsilon^2\sigma_{\text{max}}^2\lambda_1\lambda_2.
\end{align}
\clearpage   
\bibliographystyle{IEEEtran}
{\balance \bibliography{IEEEabrv,IZS}}
\end{document}